\documentclass[acmsmall]{acmart}

\usepackage{lineno}   
\usepackage{booktabs} 
\usepackage{url}  

\usepackage{multirow}  
\usepackage{color}  
\newtheorem{strategy}{Strategy}    
\newtheorem{upper bound}{Upper bound}

\usepackage{diagbox}
\usepackage{amsmath}
\usepackage{arydshln} 
\usepackage{color} 

\usepackage[linesnumbered,ruled]{algorithm2e} 

\SetAlFnt{\small}
\SetAlCapFnt{\small}
\SetAlCapNameFnt{\small}
\SetAlCapHSkip{0pt}
\IncMargin{-\parindent}

\usepackage{xspace}
\usepackage{multirow}
\usepackage{subfigure}

\setcopyright{acmcopyright}
\acmJournal{TKDD}

\begin{document}

\title{US-Rule: Discovering Utility-driven Sequential Rules}


\author{Gengsen Huang}
\affiliation{ 
	\institution{Jinan University}
	\city{Guangzhou}
	\state{Guangdong}
	\country{China}
}
\email{hgengsen@gmail.com}

\author{Wensheng Gan}
\authornote{This is the corresponding author}
\affiliation{
	\institution{Jinan University}
	\city{Guangzhou}
	\state{Guangdong}
	\country{China}
}
\email{wsgan001@gmail.com}

\author{Jian Weng}
\affiliation{
	\institution{Jinan University}
	\city{Guangzhou}
	\state{Guangdong}
	\country{China}
}
\email{cryptjweng@gmail.com}

\author{Philip S. Yu}
\affiliation{%
  \institution{University of Illinois at Chicago}
  \city{Chicago}
  \state{IL}
  \country{USA}
  \postcode{60607}
}
\email{psyu@uic.edu}


\begin{abstract}
	Utility-driven mining is an important task in data science  and has many applications in real life. High utility sequential pattern mining (HUSPM) is one kind of utility-driven mining. HUSPM aims to discover all sequential patterns with high utility. However, the existing algorithms of HUSPM can not provide an accurate probability to deal with some scenarios for prediction or recommendation. High-utility sequential rule mining (HUSRM)  was proposed to discover all sequential rules with high utility and high confidence. There is only one algorithm proposed for HUSRM, which is not enough efficient. In this paper, we propose a faster algorithm, called US-Rule, to efficiently mine high-utility sequential rules. It utilizes rule estimated utility co-occurrence pruning strategy (REUCP) to avoid meaningless computation. To improve the efficiency on dense and long sequence datasets, four tighter upper bounds (LEEU, REEU, LERSU, RERSU) and  their corresponding pruning strategies (LEEUP, REEUP, LERSUP, RERSUP) are proposed. Besides, US-Rule proposes rule estimated utility recomputing pruning strategy (REURP) to deal with sparse datasets. At last, a large number of experiments on different datasets compared to the state-of-the-art algorithm demonstrate that US-Rule can achieve better performance in terms of execution time, memory consumption and scalability.
\end{abstract}


\authorsaddresses{
Authors' addresses: 
G. Huang, W. Gan, and J. Weng, College of Cyber Security, Jinan University, Guangzhou, China; email: hgengsen@gmail.com, wsgan001@gmail.com, and cryptjweng@gmail.com. P. S. Yu, Department of Computer Science, University of Illinois at Chicago, Chicago, IL, USA; email: psyu@uic.edu.
}




\maketitle

\section{Introduction}

There are many pattern discovery techniques in data mining, and they often have lots of applications in real world. Frequent pattern mining (FPM) \cite{agrawal1995mining} is one of them and plays an import role in medical diagnosis \cite{nayak2019heart}, market basket analysis \cite{maske2018survey}, and query recommendation \cite{li2008pfp}. If the support value of a pattern meets the user-defined minimum support (\textit{minsup}) value, it is a frequent pattern. In general, FPM can discover all frequent patterns from a transaction database. These mined patterns provide analysis results for various services. Compared to FPM, sequential pattern mining (SPM) \cite{fournier2017survey,gan2019survey} considers the chronological order of items in the database. In a transaction record, the order of items generally follows the occurred order. SPM can discover sequential patterns within a certain time frame with better analysis. SPM has many applications in real life, as reviewed in \cite{fournier2017survey}, such as bioinformatics \cite{wang2007frequent}, analysis of customer consumption behavior \cite{srikant1996mining}, and web-page click event mining \cite{fournier2012using}. However, both FPM and SPM mainly consider the frequency of each items and assume that each items in one transaction event can only appear once. This limitation often not holds in many real applications. In a real transaction record, not only the quantity of commodities should be considered, but also we need to pay more attention to the price of commodities. Generally speaking, in business, we should consider profits more than the number of sales. For example, there is no doubt that the profit of $<$\textit{diamond}, \textit{ring}$>$ is far greater than the profit of $<$\textit{milk}, \textit{bread}$>$. However, as people buy milk and bread daily, the frequency of $<$\textit{milk}, \textit{bread}$>$ in the transaction database is often very large. Therefore, FPM and SPM will consider $<$\textit{milk}, \textit{bread}$>$ as a more valuable pattern. Although $<$\textit{diamond}, \textit{ring}$>$ can bring higher profits, its sales volume is too small. If \textit{minsup} is set to a larger value in FPM or SPM, the really valuable pattern $<$\textit{diamond}, \textit{ring}$>$ may be missed. For the sake of the problems posed by this situation, Ahmed \textit{et al.} \cite{ahmed2010mining} have introduced the concept of utility and proposed the problem of high-utility pattern mining (HUPM). HUPM considers internal utility and external utility of items to discover more valuable patterns. In general, the internal utility is the occurred quantity (e.g., purchased quantity) of items and the external utility is the unit utility (e.g., price) of each items. Then, the product of internal utility and external utility is the total utility of each items in the database. Different from FPM, HUPM is more challenging to find out all high-utility patterns (HUPs). The support value of a pattern is anti-monotonic, while in HUPM, if a pattern satisfies the minimum utility threshold (\textit{minutil}), the utility of its super-patterns may be smaller, equal to, or larger than the utility of itself. For example, consider patterns $<$\textit{milk}, \textit{bread}$>$ and $<$\textit{milk}, \textit{bread}, \textit{bike}$>$, the latter will be more profitable. However, the purchased number of $<$\textit{milk}, \textit{bread}$>$, as we mentioned about, is a very large value. Hence, the total profit of these two patterns is difficult to compare. After that, for further considering the order of items in a transaction record, high-utility sequential pattern mining (HUSPM) \cite{yin2012uspan} was proposed. HUSPM considers the chronological order of items in a pattern to address this issue. For $<$\textit{milk}, \textit{bread}$>$, we only know that customers purchase these shopping goods in HUPM. While in HUSPM, we can know that after buying milk, customers will continue to purchase bread. Obviously, these high-utility sequential patterns can carry more crucial and valuable information. In many real applications such as smart retail, HUSPM can utilize the chronological order of these items in a pattern to promote sales and generate high profits.


There are many algorithms about SPM \cite{han2001prefixspan, srikant1996mining, zaki2001spade} and HUSPM \cite{gan2020proum, gan2020fast, wang2014efficient, yin2012uspan} proposed to deal with many special scenarios in real-word, and they only use frequency or total utility to decide what patterns are useful. In some scenarios, if we want to predict what will happen next, SPM and HUSPM can only make a rough prediction. Even if the chronological order of items is considered in SPM and HUSPM, it can not provide accurate prediction rate for users. For instance, if we have two frequent sequential patterns $s_1$ = $<$\textit{milk}, \textit{egg}$>$ and $s_2$ = $<$\textit{milk}, \textit{egg}, \textit{bread}$>$, we can know that some customers would buy bread after buying milk and egg. However, we don't know the proportion of these customers in all customers. That is, we can only know approximately the predicted items without knowing the accuracy. Sequential rule mining (SRM) was proposed by Fournier-Viger \textit{et al.} \cite{fournier2015mining} and it uses the concept of confidence to find out all sequential rules (SRs) whose support value and confidence value satisfy the \textit{minsup} and the minimum confidence (\textit{minconf}) predefined by users. In general, SR is denoted as $\{X\}$ $\rightarrow$ $\{Y\}$. For a SR \{\textit{milk}, \textit{egg}\} $\rightarrow$ \{\textit{bread}\}, if its confidence value is 0.5, we can know that 50\% customers will buy bread after buying milk and egg. SRM is widely applied to various fields, such as business transactions, weather forecast, and financial analysis. Because the \textit{minsup} is also set in SRM, SRM also faces the problem of SPM. For example, the support value of \{\textit{milk}, \textit{egg}\} $\rightarrow$ \{\textit{bread}\} is greater than \{\textit{diamond}, \textit{ring}\} $\rightarrow$ \{\textit{rose}\}. If the confidence values of these two SR are high and satisfy the \textit{minconf}, but the \textit{minsup} is set too high, the more profitable SR \{\textit{diamond}, \textit{ring}\} $\rightarrow$ \{\textit{rose}\} may be ignored. To this end, Zida \textit{et al.} \cite{zida2015efficient} combined the concept of utility with SRM and proposed high-utility sequential rule mining (HUSRM). HUSRM considers \textit{minconf} and \textit{minutil} to generate all high-utility sequential rules (HUSRs). HUSRM can be seen as a combination of SRM and HUSPM. It can also be seen as a special version of SRM when both the external utility and internal utility are equal to one. As a result, HUSRM has many applications and extensions. Although HUSRM plays an important role than SRM in many real applications, there are few relevant studies about it. As far as we know, there is only one algorithm \cite{zida2015efficient} proposed to find out the complete set of HUSRs.

In light of some relevant studies in HUIM and HUSPM, in this paper, we propose a more efficient algorithm, named US-Rule. Inspired by the use of co-occurrence map \cite{fournier2014fast, fournier2014fhm}, we propose rule estimated utility co-occurrence map (REUCM) and rule estimated utility co-occurrence pruning strategy (REUCP) to avoid the unnecessary computations of US-Rule. There are four  more tight upper bounds introduced by US-Rule, including left expansion estimated utility (LEEU), right expansion estimated utility (REEU), left expansion reduced sequence utility (LERSU), and right expansion reduced sequence utility (RERSU). Based on these upper bounds, we design corresponding novel data structures and propose several pruning strategies. The major contributions of this paper are as follows:

\begin{itemize}
	\item A novel algorithm, called US-Rule, is proposed to discover the complete set of HUSRs. Compared to the only algorithm HUSRM, US-Rule with the use of REUCP can achieve better performance. REUCP can avoid meaningless computations of sequential rules.
	
	\item For better efficiency on dense and long sequence datasets, we propose four tighter upper bounds and design their corresponding auxiliary data structures. Based on these upper bounds, we utilize depth-first search pruning strategies (LEEUP and REEUP) and width-first search pruning strategies (LERSUP and RERSUP) to reduce the number of expansion of US-Rule. 
	
	\item To address the issue of US-Rule on sparse and short sequence datasets, we propose the rule estimated utility recomputing pruning strategy (REURP) to reduce unnecessary rule expansion and improve efficiency by recalculating the SEU of items multiple times and removing unpromising items from the database.
	
	\item Experiments using US-Rule and its different variants (with different strategies) on both real and synthetic datasets show that US-Rule with different strategies can achieve different degrees of improvement over the state-of-the-art HUSRM algorithm on different datasets.
\end{itemize}

The remainder part of this paper is organized as follows. In Section \ref{sec:relatedwork}, we review the related work about SRM, HUSPM, and HUSRM. Then, we introduce the definition of HUSRM and formalize the problem of HUSRM in Section \ref{sec:preliminaries}. The proposed  US-Rule algorithm  with several pruning strategies is presented in Section \ref{sec:algorithm}. Experimental results and evaluation on different datasets are shown and discussed in Section \ref{sec:experiments}. Finally, in Section \ref{sec:conclusion}, we conclude this paper and discuss the future work of our research.

\section{Related Work}
\label{sec:relatedwork}

\subsection{Sequential rule mining}
\label{section:SRM} 

Since frequent pattern mining (FPM) \cite{agrawal1995mining} was proposed, several sequential rule mining (SRM) algorithms have been proposed. A sequential rule can be divided into two parts, denoted as $\{X\}$ $\rightarrow$ $\{Y\}$, $X$ $\cap$ $Y$ = $\varnothing$. $X$ is the antecedent of a sequential rule, and $Y$ is the consequent. For a sequential rule, it must be satisfied the defined minimum support value (\textit{minsup}) and the minimum confidence value (\textit{minconf}). According to the total in order or partial in order, SRM can be divided to total-ordered SRM and partially-ordered SRM. The former means that all itemsets in a sequential rule depend on the order of itemsets in the original sequence. And the later indicates itemsets in antecedent or consequent of a sequential rule can be disorganized, but the order of an itemset which appears in the consequent of a sequential rule should be after the antecedent of a sequential rule. Partially-ordered SRM is the key research direction. RuleGen \cite{zaki2001spade} adopts a method of violent enumeration. Generating sequential rules from two sequential patterns. As the length of the traversed sequential pattern increases, the length of the obtained sequential rules also increases. Fournier-Viger \textit{et al.} proposed CMDeo \cite{fournier2012cmrules} and CMRules \cite{fournier2012cmrules} to improve efficiency. CMDeo is a two-phases-based algorithm. It first discovers association rules by using an association rule mining algorithm, and then these association rules are used to generate sequential rules. There is no doubt that since both RuleGen and CMDeo use the generate-and-test approach, they usually very inefficient. CMRules introduces the concept of left- and right- expansion of sequential rules, which greatly improves efficiency. The idea of left- or right- expansion through frequent items is clearly not good enough. Similar to PrefixSpan \cite{han2001prefixspan}, the RuleGrowth  algorithm  \cite{fournier2015mining} was further proposed, and it uses left- and right- expansion and pattern-growth approach for generating sequential rules, and then achieves excellent performance. To reduce the running time of SRM algorithms, an equivalence class method was proposed in ERMiner \cite{fournier2014erminer}. ERMiner continuously merges equivalence classes with the same antecedent or consequent, thus generating longer sequential rules. ERMiner also introduces a new data structure, called Sparse Count Matrix (SCM), to prune the search space by avoiding useless emergence. However, the memory consumption of ERMiner is greater than RuleGrowth.

In addition to the algorithms mentioned above, there are other popular topics in SRM. Several constraint-based SRM algorithms that can obtain valuable sequential rules are TRuleGrowth \cite{fournier2015mining}, BMSRIC-R \cite{van2021mining}, and IFERMiner \cite{thanh2020mining}. TRuleGrowth uses a sliding window to discover the sequential rules over a time period. BMSRIC-R constrains the itemset in the sequential rules to obtain more compact rules, and IFERMiner discovers the required sequential rules in a quantitative sequence database. In addition, TNS \cite{fournier2013tns} was proposed to address top-$k$ SRM that discovers top-$k$ SRs without setting \textit{minsup}. In incremental mining, IERMiner \cite{drozdyuk2020incremental} was proposed to deal with the dynamic sequence database.

\subsection{High-utility sequential pattern mining}
\label{section:HUPM}

Since Agrawal and Srikant \cite{agrawal1995mining} proposed the sequential pattern mining (SPM) to discover all frequent sequential patterns, a growing number of researchers are focusing on this area. There are many algorithms about SPM \cite{fournier2017survey,gan2019survey}, such as GSP \cite{srikant1996mining}, PrefixSpan \cite{han2001prefixspan}, SPADE \cite{zaki2001spade}, SPAM \cite{ayres2002sequential}, and non-overlapping SPM \cite{wu2021ntp}. GSP is a relatively violent method that generates longer sequential patterns by continuously merging sub-sequential patterns. It mainly utilizes the Apriori property and the support decreases with the merging of sequential patterns, thus this method is very time consuming and memory cost. In order to improve efficiency and save much memory, PrefixSpan adopts the projection mechanism to grow sequential patterns. The main idea of PrefixSpan is to avoid generating a large number of candidates by using two pruning strategies. However, PrefixSpan faces the problem of memory consumption when the projected database is too large. And then, SPADE with the equivalence class and SPAM with the bitmap both optimize the memory consumption problem. Nevertheless, the main problem of SPM is the use of support to mine sequential patterns. The researchers then introduced utility into SPM and proposed a new framework called utility mining \cite{gan2019surveyU} as well as many utility-related calculation methods. The goal of high-utility sequential pattern mining (HUSPM) \cite{gan2019surveyU} is to find out all high-utility sequential patterns (HUSPs) whose total utility is not less than the predefined minimum utility (\textit{minutil}). The concept of utility takes into account some theories of economics (e.g., profit, risk, and importance).

In SPM, many algorithms can utilize the Apriori property to prune unpromising sequential patterns. Unfortunately, this property makes no effect on HUSPM. Due to the use of utility, the analysis of HUSPM is much more difficult. Utility is a not anti-monotonic and not monotonic value \cite{song2016high,gan2019surveyU}. To address this problem, Ahmed \textit{et al.} \cite{ahmed2010mining} introduced sequence weighted utilization (SWU) as a downward closure property. SWU is the first upper bound about HUSPM. Similar to the Apriori property, we can utilize SWU to reduce search space. And then, UtilityLevel (an Apriori-like algorithm) and UtilitySpan (an algorithm based on PrefixSpan) were proposed by Ahmed \textit{et al.} \cite{ahmed2010novel}. Both UtilityLevel and UtilitySpan are two-phases-based algorithms. They use SWU to discover all candidate sequences at first, then extract real HUSPs according to \textit{minutil}. Next, inspired by the idea of tree, UWAS-tree \cite{ahmed2010mining} and IUWAS-tree \cite{ahmed2010mining} are proposed to find HUSPs from web log data. In addition, Shie \textit{et al.} \cite{shie2011mining} designed the UMSP algorithm for analyzing mobile data. However, these algorithms face the problem that too much memory consumption is caused by a large number of candidate patterns. What's more, they only consider single-item sequences whose itemsets only contain one item. In order to deal with multiple-item-based sequences, Yin \textit{et al.} \cite{yin2012uspan} proposed a novel algorithm, called USpan. USpan can discover multiple-item sequences through its utility-matrix structure. Thus the order of items in itemsets can be considered, and the discovered HUSPs are more valuable. Besides, USpan also utilizes SWU and sequence projected utilization (SPU) in two designed depth  and width pruning strategies to prune search space. Both the pruning strategies can be applied into the lexicographic tree to improve efficiency. Unfortunately,  USpan not only will generate too many promising sequences with a high SWU value, but also will miss some real HUSPs due to the use of SPU. Alkan \textit{et al.} \cite{alkan2015crom} then designed HuspExt for better efficiency. Lan \textit{et al.} \cite{lan2014applying} also introduced a projection-based algorithm namely PHUS to improve the mining performance. PHUS utilizes a sequence-utility upper-bound model (SUUB) to discover all HUSPs, and SUUB introduces the maximum utility measure to calculate the more accurate utility of a sequence. Wang \textit{et al.} \cite{wang2016efficiently} proposed HUS-Span for further improve efficiency. HUS-Span introduces a new upper bound, called the prefix extension utility (PEU). HUS-Span discovers all HUSPs, but it's efficiency is not very good. A large number of candidate sequences will cause too much memory in HUS-Span. Recently, Gan \textit{et al.} \cite{gan2020proum} proposed a projection-based utility mining algorithm, called ProUM. ProUM can quickly grow sequences through a new data structure, called utility-array. Besides, the sequence extension utility (SEU) and its corresponding pruning strategies are also introduced in ProUM. Subsequently, Gan \textit{et al.} \cite{gan2020fast} proposed a faster algorithm, called HUSP-ULL, which utilizes a new data structure, called UL-list, to quick generate projected sub-databases. Besides, two powerful pruning strategies are also introduced by HUSP-ULL.

There are many fast algorithms have been proposed to discover HUSPs. In recent years, researchers have also extended HUSPM to other fields. TKHUS-Span \cite{wang2016efficiently} and TKUS \cite{zhang2021tkus} are top-$k$ HUSPM algorithms that can find out top-$k$ high-utility sequences without setting \textit{minutil}.  Recently, the TUSQ approach \cite{zhang2021tusq} was proposed to deal with a new problem of utility-driven targeted sequence querying, and OSUMS \cite{zhang2021shelf} was introduced for on-shelf utility mining (OSUM). HAOP-Miner \cite{wu2021haop}  aims at mining high-average utility one-off sequential patterns.

\subsection{High-utility sequential rule mining}
\label{section:HUSRM}

Although many algorithms for SRM have been proposed and SRM has been studied in many extensions, there is only one algorithm makes use of utility to discover high-utility sequential rules. Different from SRM, HUSRM assumes that the input sequences have rich utility factor but do not have duplicate items. To the best of our knowledge, Zida \textit{et al. }\cite{zida2015efficient} first defined the concept of HUSRM and proposed the first algorithm called HUSRM. HUSRM introduces the \textit{minutil} as a measure to judge a sequential rule whether is a high utility sequential rule (HUSR). Based on one data structure called utility-table and several optimizations, the HUSRM algorithm can efficiently mine all HUSRs. The information of each rule is stored in its utility-table, thus HUSRM can accelerate the process of utility calculation of sequential rules. Besides, five optimizations are used to improve efficiency. The first optimization is to remove the unpromising items at an early stage, and this can avoid useless items be a part of a rules. The second one is to remove the unpromising rules at an early stage, then reduce the count of expansion. The third one saves memory consumption by the use of bitmap. The fourth one also can save memory by the use of more compact utility-table for left expansion. The last one utilizes tighter upper bounds as expansion condition to improve efficiency. After that, several extensions of HUSRM have been studied. For example, Zhang \textit{et al.} \cite{zhang2020hunsr} incorporated the concept of negative sequence  into HUSRM and proposed a solution to find out high-utility negative sequential rules.

\section{Preliminaries and Problem Formulation}
\label{sec:preliminaries}

This section introduces fundamental concepts and notations used in this paper. And then, we formalize the problem definition of high-utility sequential rule mining.

\subsection{Notations and concepts}

\begin{definition}[Sequence database]
	\rm Let $I$ = \{$i_1$, $i_2$, $\cdots$, $i_l$\} be a set of $l$ distinct items. An item is denoted as $(i, q(i))$, and $q(i)$ is its internal utility. An itemset (element) $X$ is a subset of $I$, and $X$ = \{$i_1$, $i_2$, $\cdots$, $i_m$\} $\subseteq$ $I$. We assume that all distinct items in $X$ are ordered according to $\succ_{lex}$, which is a lexicographical order $(a$ $\textless$ $b$ $\textless$ $\cdots$ $\textless$ $z)$. For a sequence $s$, it is composed of itemsets and can be denoted as $s$ = $<$$I_1$, $I_2$, $\cdots$, $I_n$$>$, $I_k$ $\subseteq$ $I$ $(1 \le k \le n)$. A sequence database $\mathcal{D}$ consist of all sequences, $\mathcal{D}$ = $<$$s_1$, $s_2$, $\cdots$, $s_p$$>$, and each sequence own its unique identifier (\textit{SID}). Like HUSRM \cite{zida2015efficient}, we assume that an item cannot appear in one sequence more times. Each different item has its corresponding external utility, denoted as $eu(i)$. Note that both internal utility and external utility are positive.
\end{definition}

\begin{table}[h]
	\centering
	\caption{Sequence database}
	\label{table1}
	\begin{tabular}{|c|c|}  
		\hline 
		\textbf{SID} & \textbf{Sequence} \\
		\hline  
		\(s_{1}\) & $<$$(a, 1)$ $(b, 2)$ $\{(c, 1) (g, 1)\}$$>$ \\ 
		\hline
		\(s_{2}\) & $<$$(a, 1) (c, 1)$ $\{(e, 1) (g, 1)\}$$>$ \\  
		\hline  
		\(s_{3}\) & $<$$(b, 1)$ $(d, 1)$ $(g, 2)$$>$ \\
		\hline  
		\(s_{4}\) & $<$$(e, 1)$ $(f, 2)$$>$ \\
		\hline
	\end{tabular}
\end{table}

\begin{table}[!h]
	\caption{External utility table}
	\label{table2}
	\centering
	\begin{tabular}{|c|c|c|c|c|c|c|c|}
		\hline
		\textbf{Item}	    & $a$	& $b$	& $c$	& $d$	& $e$	& $f$  & $g$ \\ \hline 
		\textbf{Unit utility}	    & 2     & 1     & 4     & 7     & 1     & 2    & 1 \\ \hline
		
	\end{tabular}
\end{table}

As shown in Tables \ref{table1} and \ref{table2}, in this paper, the example sequence database has five sequences. We use $s_1$, $s_2$, $s_3$, $s_4$, and $s_5$ to present their \textit{SID}. Each item is a single letter, owning its internal utility in the itemset and is ordered according to $\succ_{lex}$. If an itemset contain only one item, we omit curly brackets to present it. For each item, Table \ref{table2} shows their external utilities. For example, given an item $(a,1)$, we can know that $q(a)$ = 1 and $eu(a)$ = 2.

\begin{definition}[Sequential rule and sequential rule size]
	\label{definition:SR}
	\rm A sequential rule is denoted as $r$ = $X$ $\rightarrow$ $Y$, and $X$ is the antecedent of sequential rule $r$ while $Y$ is the consequent. Both $X$ and $Y$ are not empty and unordered itemsets, and $X$ $\cap$ $Y$ = $\varnothing$. For partially-ordered SRM, if $X$ appears in a sequence, $Y$ will occur after $X$. Note that we stipulate that the items in a sequence appear no more than once. Therefore, we argue that $k$ is the number of distinct items in the antecedent of $r$, and $m$ is the number of distinct items in the consequent of $r$. Then, the size of $r$ can denoted as $k$ $\ast$ $m$. Given another sequential rule $r_1$ and its size is $f$ $\ast$ $h$, if the size of $r_1$ is larger than $r$, it must be $f$ $\textgreater$ $k$ and $h$ $\ge$ $m$, or $f$ $\ge$ $k$ and $h$ $\textgreater$ $m$.
\end{definition}

For example, as shown in Table \ref{table1}, $r_1$ = $\{a, b\}$ $\rightarrow$ $\{c\}$ is a sequential rule, and its size is 2 $\ast$ 1. Given another sequential rule $r_2$ = $\{a, b\}$ $\rightarrow$ $\{c, g\}$, we can know that $r_2$ is larger than $r_1$. It is because the size of $r_2$ is 2 $\ast$ 2.

\begin{definition}[Support and confidence]
	\rm We use \textit{seq}$(r)$ to present the sequence containing the sequential rule $r$. For a sequential rule $r$ = $X$ $\rightarrow$ $Y$, its support value is the number of sequences in $\mathcal{D}$ whose sub-sequences contain \textit{seq}$(r)$. And it can defined as \textit{sup}$(r)$ = \textit{seq}$(r)$ / $\lvert$$\mathcal{D}$$\rvert$. While the confidence of rule $r$ is the proportion of $Y$ occurs after $X$ and it can be defined as \textit{conf}$(r)$ = \textit{sup}$(r)$ $/$ \textit{sup}$(X)$.
\end{definition}

For example, from Table \ref{table1}, we can clearly see that the support of $r$ = $\{a\}$ $\rightarrow$ $\{c, g\}$ is 2, because $r$ appears in sequences $s_1$ and $s_2$. And the confidence of $r$ is 1, because $\{a\}$ appears in sequences $s_1$ and $s_2$, and $\{a, c, g\}$ also appears in these two sequences. Therefore, \textit{sup}$(r)$ = 2, and \textit{conf}$(r)$ = 1.

\begin{definition}[Utility of an item/itemset in a sequence]	
	\rm We specify that for an item in a sequence, its utility value is equal to the product of its internal utility and external utility. It can denoted as $u(i, s_l)$ and defined as $u(i, s_l)$ = $q(i,s_l)$ $\times$ $eu(i)$. Similarly, for an itemset $I$, its utility in a sequence is the sum of the utility of all items in this itemset. It can denoted as $u(I, s_l)$ and defined as $u(I, s_l)$ = $\sum_{i \in I}$ $q(i,s_l)$ $\times$ $eu(i)$.
\end{definition}

\begin{definition}[Utility of an item/itemset in a database]	
	\rm As for the utility of an item $i$ in a database $\mathcal{D}$ is the sum of utilities in all sequences. It can be denoted as $u(i)$ and defined as $u(i)$ = $\sum_{s_l \in \mathcal{D}}$ $q(i, s_l)$ $\times$ $eu(i)$. For an itemset $I$, its utility in a database $\mathcal{D}$ is the total utility of all itemsets in this database. It can be denoted as $u(I)$ and defined as $u(I)$ = $\sum_{i \in I \land I \subseteq s_l \land s_l \in \mathcal{D}}$ $q(i,s_l)$ $\times$ $eu(i)$.
\end{definition}

For example, in Table \ref{table1}, the utility of item $(a,1)$ in sequence $s_1$ = $q(a, s_1)$ $\times$ $eu(a)$ = 2, and the utility of itemset $\{(c,1) (g,1)\}$ = $q(c,s_1)$ $\times$ $eu(c)$ + $q(g,s_1)$ $\times$ $eu(g)$ = 4 + 1 = 5.  $u(a)$ = $q(a,s_1)$ $\times$ $eu(a)$ + $q(a,s_2)$ $\times$ $eu(a)$ = 2 + 2 = 4, and $u((c,g))$ = $q(c,s_1)$ $\times$ $eu(c)$ + $q(g,s_1)$ $\times$ $eu(g)$ = 4 + 1 = 5. We can know that the total utility of $a$ in $\mathcal{D}$ is 4, and the total utility of $(c,g)$ in $\mathcal{D}$ is 5.

\begin{definition}[Utility of a sequential rule in a sequence and database]	
	\rm The utility of a sequential rule $r$ in sequence $s_l$ can be denoted as $u(r, s_l)$ and defined as $u(r, s_l)$ = $\sum_{i \in r \land \textit{seq}(r) \subseteq s_l }$ $q(i, s_l)$ $\times$ $eu(i)$. The utility of a sequential rule $r$ in database $\mathcal{D}$ can be denoted as $u(r)$ and defined as $u(r)$ = $\sum_{i \in r \land \textit{seq}(r) \subseteq s_l \land s_l \in \mathcal{D}}$ $q(i, s_l)$ $\times$ $eu(i)$.
\end{definition}

For example, in Table \ref{table1}, the utility of sequential rule $r$ = $\{a\}$ $\rightarrow$ $g$ is $u(r)$ = $u(r,s_1)$ + $u(r,s_2)$ = 3 + 3 = 6, it is because $r$ appears in $s_1$ and $s_2$.

\subsection{Problem statement}

\begin{definition}[High utility sequential rule mining]	
	\rm We assume that the \textit{minutil} value is positive and the \textit{minconf} $\in$ $[0, 1]$, and they are predefined by users. Given these two thresholds and a sequence database $\mathcal{D}$, if a sequential rule $r$ is satisfied $u(r)$ $\ge$ \textit{minutil} and \textit{conf}$(r)$ $\ge$ \textit{minconf} simultaneously, we call it a high-utility sequential rule (HUSR). Otherwise, it is an invalid rule. The problem of high-utility sequential rule mining (HUSRM) is to discover a complete set of all HUSRs from a sequence database.
\end{definition}

\begin{table}[!htbp]
	\centering
	\caption{HUSRs in Table \ref{table1} when \textit{minconf} = 0.5 and \textit{minutil} = 10}
	\label{table3}
	
	\begin{tabular}{|c|c|c|c|c|}  
		\hline 
		\textbf{ID} & \textbf{HUSR} & \textbf{Support} & \textbf{Confidence} & \textbf{Utility}\\
		\hline 
		\(r_{1}\) & $\{a\}$ $ \rightarrow $ $\{c\}$ & 2 & 1.0 & 12\\ 
		\hline 
		\(r_{2}\) & $\{a\}$ $ \rightarrow $ $\{c, g\}$ & 2 & 1.0 & 14\\
		\hline 
		\(r_{3}\) & $\{b\}$ $ \rightarrow $ $\{d, g\}$ & 1 & 0.5 & 10\\
		\hline 
		\(r_{4}\) & $\{b, d\}$ $ \rightarrow $ $\{g\}$ & 1 & 1.0 & 10\\
		\hline
	\end{tabular}
\end{table}

Table \ref{table3} shows the discovered HUSRs under the setting of \textit{minutil} = 10 and \textit{minconf} = 0.5. From the results, we can see that totally four HUSRs are found, and they both satisfy the two predefined thresholds. Although $r$ = $\{b\}$ $\rightarrow$ $\{g\}$ has a high confidence and its support is 2, it is removed due to its low utility. Obviously, the results of HUSRM are very different from those of SRM.

\begin{definition}[The expansion of a sequential rule]
	\label{ESR}
	\rm US-Rule utilizes the approach similar to RuleGrowth \cite{fournier2015mining} to grow a sequential rule $r$ = $\{X\}$ $\rightarrow$ $\{Y\}$. This means that US-Rule also uses the left and right expansion. For an expanded item $i$ $\in$ $I$, the left expansion of $r$ is defined as $X$ $\cup$ $\{i\}$ $\rightarrow$ $Y$, and $i$ must be larger than all item in $X$ according to $\succ_{lex}$ and $i$ $\notin$ $Y$. While the right expansion of $r$ is defined as $X$ $\rightarrow$ $Y$ $\cup$ $\{i\}$, and $i$ must be larger than all item in $Y$ according to $\succ_{lex}$ and $i$ $\notin$ $X$. Let $r^\prime$ be a new rule expanded from $r$, the support of $r^\prime$ can be lower or equal to \textit{sup}$(r)$. As for the effect on confidence, the utility of $r^\prime$ can be lower, higher, or equal to \textit{conf}$(r)$. Similarly, the utility of $r^\prime$ can be lower, higher, or equal to $u(r)$. Relevant proofs can be referred to Ref. \cite{fournier2015mining, zida2015efficient}. In order to avoid generating the same sequential rule twice, US-Rule utilizes the following regulation: a sequential rule cannot perform the right expansion after it performs the left expansion.
\end{definition}

\section{The US-Rule Algorithm}
\label{sec:algorithm}

In this section, we introduce the key definitions and new upper bounds, and then present our designed data structures and pruning strategies. Finally, we present the complete pseudocode of US-Rule which is an one-phase rule-growth-based algorithm.

\subsection{Definition and upper bounds}

Note that the US-Rule algorithm adopts some concepts from previous studies in utility-based sequence mining to reduce the search space. Left expansion estimated utility (LEEU) and right expansion estimated utility (REEU) are inspired by prefix estimated utility (PEU), while left expansion reduced sequence utility (LERSU) and right expansion reduced sequence utility (RERSU) are inspired by reduced sequence utility (RSU). In addition to using the four new upper bounds, we also use the sequence estimated utility (SEU). The details can be referred to Ref. \cite{wang2016efficiently, yin2012uspan}.

\begin{definition}[Sequence utility]
	\rm For a sequence $s_l$, its sequence utility is the total utility of all items in this sequence, and it can be defined as \textit{SU}$(s_l)$ = $\sum_{i \in s_l}$ $q(i, s_l)$ $\times$ $p(i)$. 
\end{definition}

For instance, in Table \ref{table1}, the utility of sequence $s_1$, $s_2$, $s_3$, and $s_4$ are 9, 8, 10, and 5, respectively.

\begin{definition}[Sequence estimated utility of item/rule]
	\rm For an item $i$, its sequence estimated utility (SEU) is the sum of the utility of sequences containing $i$ and it can be defined as \textit{SEU}$(i)$ = $\sum_{i \in s_l \land s_l \in \mathcal{D}}$ \textit{SU}$(s_l)$. For a sequential rule $r$, its SEU is the SEU of \textit{seq}$(r)$ and it can be defined as \textit{SEU}$(r)$ = $\sum_{\textit{seq}(r) \subseteq s_l \land s_l \in \mathcal{D}}$ \textit{SU}$(s_l)$. Then \textit{SEU}$(i)$ $\ge$ $u(i)$ and \textit{SEU}$(r)$ $\ge$ $u(r)$.
\end{definition}

For example, in Table \ref{table1}, the SEU of $a$ is \textit{SEU}$(a)$ = 9 + 8 = 17, because $a$ appears in sequence $s_1$ and sequence $s_2$. Also the SEU of sequential rule $r$ = $\{e\}$ $\rightarrow$ $\{f\}$ is \textit{SEU}$(r)$ = 5, because only sequence $s_4$ contains $r$.

\begin{definition}[Promising item and promising rule]
	\rm If \textit{SEU}$(i)$ $\ge$ \textit{minutil}, the item $i$ is a promising item. Otherwise, it is an unpromising item. For a sequential rule $r$, if \textit{SEU}$(r)$ $\ge$ \textit{minutil}, $r$ is a promising rule. Otherwise, it is an unpromising rule.  
\end{definition}

For instance, if \textit{minutil} = 10, we can see that $f$ is an unpromising item, because \textit{SEU}$(f)$ = 5 $\textless$ 10. For a sequential rule $r$ = $\{e\}$ $\rightarrow$ $\{f\}$, $r$ is an unpromising rule, because \textit{SEU}$(r)$ = 5 $\textless$ 10.

\begin{strategy}[Unpromising items pruning strategy]
	\label{Strategy:UIP}
	\rm Given an item $i$, if $i$ is an unpromising item, US-Rule will remove it from the database. The reason is that if \textit{SEU}$(i)$ $\textless$ \textit{minutil}, the utility of a rule contains $i$ will be smaller than \textit{minutil}.
\end{strategy}

\begin{strategy}[Unpromising sequential rules pruning strategy]
	\label{Strategy:USRP}
	\rm Given a sequential rule $r$, according to unpromising sequential rules pruning strategy (USRP), if $r$ is an unpromising rule, US-Rule will terminate its expansion. The reason is that if \textit{SEU}$(r)$ $\textless$ \textit{minutil}, the utility of the expansion of a rule will be smaller than \textit{minutil}. 
\end{strategy}

Note that both Strategy \ref{Strategy:UIP} and Strategy \ref{Strategy:USRP} can be referred to Ref. \cite{zida2015efficient}.

\begin{definition}[Rule estimated utility co-occurrence pruning map]
	\rm We propose a novel data structure, called rule estimated utility co-occurrence pruning map (REUCM) to store the SEU of item $a$ and item $b$. Item $a$ must be in front of item $b$. Then there will be two situations. The first one is that the itemset of item $a$ precedes the itemset of item $b$. The second one is that item $a$ and  item $b$ appear in the same itemset. In general, using a matrix for storage will consume a lot of memory and there are many zeros in the matrix. Therefore, US-Rule utilizes a hash-table for the design. 
\end{definition}

For instance, the SEU of item $a$ and item $b$ can be denoted as \textit{REUCM}$(a, b)$ and defined as \textit{REUCM}$(a, b)$ = $\sum_{(\textit{seq}(\{a\} \rightarrow \{b\}) \subseteq s \; or \; (a,b) \subseteq s) \land s \in \mathcal{D} \land s \in \mathcal{D}}$ $u(s)$. Note that it is different from the estimated utility co-occurrence structure (EUCS) \cite{fournier2014fhm}. In EUCS, the order of item $a$ and item $b$ does not need to be considered. While in REUCM, we need to consider this. For \textit{REUCM}$(a, b)$, it is not equal to \textit{REUCM}$(b, a)$. It is because $\{a\}$ $\rightarrow$ $\{b\}$ is different from $\{b\}$ $\rightarrow$ $\{a\}$. Note that \textit{REUCM}$(a, b)$ $\ge$ \textit{SEU}$(\{a\}$ $\rightarrow$ $\{b\})$.

\begin{theorem}
	Given two items $a$ and $b$ ($a$ in front of $b$), and a sequential rule $r$ = $\{a\}$ $\rightarrow$ $\{b\}$, \textit{REUCM}$(a, b)$ $\ge$ \textit{SEU}$(r)$.
\end{theorem}

\begin{proof}
	\rm We can know that item $a$ appears in front of item $b$. Then there will be two situations. The first one is that the itemset of item $a$ precedes the itemset of item $b$. The second one is that item $a$ and item $b$ appear in the same itemset. For the first case, its estimated utility can be expressed as \textit{SEU}$_1$ = \textit{SEU}$(\{a\}$ $\rightarrow$ $\{b\})$. For another case, the estimated utility of it can be expressed as \textit{SEU}$_2$ = \textit{SEU}(\textit{seq}$((a,b)))$ and \textit{SEU}$(r)$ = 0. Therefore, \textit{REUCM}$(a, b)$ = \textit{SEU}$_1$ + \textit{SEU}$_2$ $\ge$ \textit{SEU}$(r)$.
\end{proof}

\begin{strategy}[Rule estimated utility co-occurrence pruning strategy]
	\rm Inspired by EUCP \cite{fournier2014fhm}, US-Rule proposes the rule estimated utility co-occurrence pruning strategy (REUCP) to void meaningless operation in expansion. Relied on the novel data structure of REUCM, given a sequential rule $r$ = $\{X\}$ $\rightarrow$ $\{Y\}$, item $m$ and item $n$ respectively are the largest item of $X$ and $Y$ according to $\succ_{lex}$. For an item $i$, if \textit{REUCM}$(i, n)$ $\textless$ \textit{minutil}, item $i$ will not perform this left expansion; and if \textit{REUCM}$(m, i)$ $\textless$ \textit{minutil}, item $i$ will not perform this right expansion.
\end{strategy}

\begin{proof}
	\rm Given a sequential rule $r$ = $\{X\}$ $\rightarrow$ $\{Y\}$ and an item $a$ that can only be left extended. We assume that the sequential rule $r^\prime$ = $\{X\} \cup a$ $\rightarrow$ $\{Y\}$ is generated form $r$ by performing left expansion on $a$. Let $S_r$ be the sequences set of \textit{seq}$(r)$, and $S_{r^{\prime}}$ to denote the sequences set of \textit{seq}$(r^{\prime})$. Since \textit{seq}$(r)$ $\subset$ \textit{seq}$(r^{\prime})$, we can know that $S_{r^{\prime}}$ $\subseteq$ $S_r$. As discussed in Definition \ref{definition:SR}, let $n$ be the largest item in $Y$, and $a$ will be front of $n$. And let $t$ denote $\{a\} \rightarrow \{n\}$ and  $S_t$ denote the set of all sequences which contain \textit{seq}$(t)$. Because \textit{seq}$(t)$ $\subset$ \textit{seq}$(r)$ $\subset$ \textit{seq}$(r^{\prime})$,  $S_{r^{\prime}}$ $\subseteq$ $S_r$ $\subseteq$ $S_t$. And we can know that \textit{SEU}$(t)$ $\ge$ \textit{SEU}$(r)$ $\ge$ \textit{SEU}$(r^{\prime})$. If \textit{REUCM}$(a, n)$ $\textless$ \textit{minutil}, there will be \textit{SEU}($r^{\prime})$ $\le$ \textit{SEU}$(r)$ $\le$ \textit{SEU}$(t)$ $\le$ \textit{REUCM}$(a, n)$ $\textless$ \textit{minutil}. Therefore, $r^{\prime}$ is an unpromising rule, we can terminate this expansion. The proof about the right expansion is similar to this proof.
\end{proof}

\begin{upper bound}[Left expansion estimated utility]
	\rm The left expansion estimated utility (LEEU) of a sequential rule $r$ in sequence $s$, denoted as \textit{LEEU}$(r, s)$, is defined as:
	$$ \textit{LEEU}(r,s)=\left\{
	\begin{aligned}
	u(r,s) + \textit{ULeft}(r,s) + \textit{ULeftRight}(r,s) &  & \textit{ULeft}(r,s) + \textit{ULeftRight}(r,s) \textgreater 0 \\
	0 & & otherwise
	\end{aligned}
	\right.
	$$
\end{upper bound}

Here \textit{ULeft}$(r,s)$ denotes the total utility of all items in sequence $s$ which can only extend sequential rule $r$ by left expansion and \textit{ULeftRight}$(r,s)$ denotes the total utility of all items in sequence $s$ which both can extend sequential rule $r$ by left expansion and right expansion. The left expansion estimated utility of a sequential rule $r$ in sequence database $\mathcal{D}$, denoted as \textit{LEEU}$(r)$, is defined as:
\begin{center}
	\textit{LEEU}$(r)$ = $\sum_{\textit{seq}(r)\subseteq s_l \land s_l \in \mathcal{D}}$ \textit{LEEU}$(r, s_l)$
\end{center}

Consider the sequential rule $r$ = $\{a\}$ $\rightarrow$ $\{c\}$, $r$ appears in sequence $s_1$ and $s_2$. In sequence $s_1$, item $b$ can only extend $r$ by left expansion, \textit{ULeft}$(r, s_1)$ = $u(b)$ = 2, and \textit{LEEU}$(r, s_1)$ will be equal to $u(r)$ + \textit{ULeft}$(r,s_1)$ + \textit{ULeftRight}$(r,s_1)$ = 6 + 2 + 0 = 8. This is because there is no item can extend $r$ in sequence $r_2$, \textit{LEEU}$(r, s_2)$ = 0. Therefore, \textit{LEEU}$(r)$ = \textit{LEEU}$(r, s_1)$ + \textit{LEEU}$(r, s_2)$ = 8 + 0 = 8.

\begin{theorem}
	\label{Theorem:LEEU}
	Given a sequential rule $r$ = $\{X\}$ $\rightarrow$ $\{Y\}$ and its left expansion $r^\prime$ = $\{X\}$ $\cup$ $a$ $\rightarrow$ $\{Y\}$, $u(r^\prime)$ $\le$ \textit{LEEU}$(r)$. 
\end{theorem}

\begin{proof}
	\label{Proof:LEEU}
	\rm In sequence $s$, since the item $a$ can be an item of left expansion for $r$, then $u(a,s)$ $\le$ \textit{ULeft}$(r,s)$ + \textit{ULeftRight}$(r,s)$. And we can know that $u(r^\prime,s)$ = $u(r,s)$ + $u(a,s)$ $\le$ $u(r,s)$ + \textit{ULeft}$(r,s)$ + \textit{ULeftRight}$(r,s)$ = \textit{LEEU}$(r,s)$. Therefore, $u(r^\prime)$ $\le$ \textit{LEEU}$(r)$.
\end{proof}

\begin{upper bound}[Right expansion estimated utility]
	\rm The right expansion estimated utility (REEU) of a sequential rule $r$ in sequence $s$, denoted as \textit{REEU}$(r, s)$, is defined as:
	$$ \textit{REEU}(r,s)=\left\{
	\begin{aligned}
	u(r,s) + \textit{UtilityExtend}(r,s)  &  & \textit{UtilityExtend}(r,s)  \textgreater 0 \\
	0 & & otherwise
	\end{aligned}
	\right.
	$$
\end{upper bound}

We use \textit{URight}$(r,s)$ to denote the total utility of all items in sequence $s$ which can only extend sequential rule $r$ by right expansion and \textit{UtilityExtend}$(r,s)$ to denote the total utility of all items which can extend sequential rule $r$ by left expansion or right expansion. We can know that \textit{UtilityExtend}$(r,s)$ = \textit{ULeft}$(r,s)$ + \textit{ULeftRight}$(r,s)$ + \textit{URight}$(r,s)$. Different from the equation of \textit{LEEU}, \textit{UtilityExtend}$(r,s)$ is a larger value. This is because the right extension can be followed by the left extension, and the right extension needs to take into account that the left extension may make the utility of a sequential rule larger.  

The right expansion estimated utility of a sequential rule $r$ in database $\mathcal{D}$, denoted as \textit{REEU}$(r)$, is defined as:
\begin{center}
	\textit{REEU}$(r)$ = $\sum_{\textit{seq}(r)\subseteq s_l \land s_l \in \mathcal{D}}$ $\textit{REEU}(r, s_l)$
\end{center}

Consider the sequential rule $r$ = $\{a\}$ $\rightarrow$ $\{c\}$, where $r$ appears in $s_1$ and $s_2$. In sequence $s_1$, item $b$ and item $g$ can only extend $r$, because \textit{UtilityExtend}$(r, s_1)$ = $u(b)$ + $u(g)$ = 2 + 1 = 3. \textit{REEU}$(r, s_1)$ will be equal to $u(r)$ + \textit{UtilityExtend}$(r,s_1)$ = 6 + 3 = 9. The items $e$ and $g$ can extend $r$ in sequence $r_2$, thus \textit{REEU}$(r, s_2)$ = 8. Therefore, \textit{REEU}$(r)$ = \textit{REEU}$(r, s_1)$ + \textit{REEU}$(r, s_2)$ = 9 + 8 = 17.

\begin{theorem}
	\label{Theorem:REEU}
	Given a sequential rule $r$ = $\{X\}$ $\rightarrow$ $\{Y\}$ and its right expansion $r^\prime$ = $\{X\} $ $\rightarrow$ $\{Y\} \cup a$, $u(r^\prime)$ $\le$ \textit{REEU}$(r)$. 
\end{theorem}

\begin{proof}
	\rm In a sequence $s$, since the item $a$ can be an item of right expansion for $r$, then $u(a,s)$ $\le$ \textit{UtilityExtend}$(r,s)$. And we can know that $u(r^\prime,s)$ = $u(r,s)$ + $u(a,s)$ $\le$ $u(r,s)$ + \textit{UtilityExtend}$(r,s)$ = \textit{REEU}$(r,s)$. Therefore, $u(r^\prime)$ $\le$ \textit{REEU}$(r)$.
\end{proof}

\begin{upper bound}[Left expansion reduced sequence utility]
	\rm Let $\alpha$ be a sequential rule able to generate sequential rule $r$ by left expansion. The left expansion reduced sequence utility (LERSU) of $r$ in sequence $s$, can be denoted as \textit{LERSU}$(r, s)$ and defined as:
	$$ \textit{LERSU}(r,s)=\left\{
	\begin{aligned}
	\textit{LEEU}(\alpha,s)   &  & \textit{seq}(r) \subseteq s \land \textit{seq}(\alpha) \subseteq s \\
	0 & & otherwise
	\end{aligned}
	\right.
	$$
\end{upper bound}

The left expansion reduced sequence utility of a sequential rule $r$ in database $\mathcal{D}$, denoted as \textit{LERSU}$(r)$, is defined as:
\begin{center}
	\textit{LERSU}$(r)$ = $\sum_{\forall s \in \mathcal{D}}$ \textit{LERSU}$(r, s)$
\end{center}

Consider the sequential rule $r$ = $\{a, c\}$ $\rightarrow$ $\{g\}$, $r$ appears in sequence $s_2$, and it generated from $\alpha$ = $\{a\}$ $\rightarrow$ $\{g\}$ by performing the left expansion. Therefore, \textit{LERSU}$(r)$ = \textit{LEEU}$(\alpha, s_2)$ = 8.

\begin{theorem}
	\label{Theorem:LERSU}
	Given a sequential rule $r$ = $\{X\}$ $\rightarrow$ $\{Y\}$ and its left expansion $r^\prime$ = $\{X\} \cup a$ $\rightarrow$ $\{Y\} $, $u(r^\prime)$ $\le$ \textit{LERSU}$(r)$. 
\end{theorem}

\begin{proof}
	\rm Let $\alpha$ be a sequential rule able to generate sequential rule $r$ by left expansion. In a sequence $s$, \textit{LERSU}$(r,s)$ = \textit{LEEU}$(\alpha,s)$. Based on the Theorem \ref{Theorem:LEEU}, we can know that $u(r^\prime, s)$ $\le$ \textit{LEEU}$(r,s)$ $\le$ \textit{LEEU}$(\alpha,s)$ = \textit{LERSU}$(r,s)$. Therefore, we have $u(r^\prime)$ $\le$ \textit{LERSU}$(r)$.
\end{proof}

\begin{upper bound}[Right expansion reduced sequence utility]
	\rm Let $\alpha$ be a sequential rule able to generate sequential rule $r$ by right expansion. The right expansion reduced sequence utility (RERSU) of $r$ in sequence $s$, denoted as \textit{RERSU}$(r, s)$, is defined as:
	$$ \textit{RERSU}(r,s)=\left\{
	\begin{aligned}
	\textit{REEU}(\alpha,s)   &  & \textit{seq}(r) \subseteq s \land \textit{seq}(r^\prime) \subseteq s \\
	0 & & otherwise
	\end{aligned}
	\right.
	$$
\end{upper bound}

The right expansion reduced sequence utility of a sequential rule $r$ in database $\mathcal{D}$, denoted as \textit{RERSU}$(r)$, is defined as:
\begin{center}
	\textit{RERSU}$(r)$ = $\sum_{\forall s \in \mathcal{D}}$ \textit{RERSU}$(r, s)$
\end{center}

Consider the sequential rule $r$ = $\{a\}$ $\rightarrow$ $\{c, g\}$, $r$ appears in sequence $s_1$ and $s_2$, and it generated from $\alpha$ = $\{a\}$ $\rightarrow$ $\{c\}$ by performing the right expansion. Therefore, \textit{RERSU}$(r)$ = \textit{REEU}$(\alpha, s_1)$ + \textit{REEU}$(\alpha, s_2)$ = 9 + 8 = 17.

\begin{theorem}
	\label{Theorem:RERSU}
	Given a sequential rule $r$ = $\{X\}$ $\rightarrow$ $\{Y\}$ and its right expansion $r^\prime$ = $\{X\} $ $\rightarrow$ $\{Y\} \cup a $, $u(r^\prime)$ $\le$ \textit{RERSU}$(r)$. 
\end{theorem}

\begin{proof}
	\rm Let $\alpha$ be a sequential rule able to generate sequential rule $r$ by the right expansion. In a sequence $s$, \textit{RERSU}$(r,s)$ = \textit{REEU}$(\alpha,s)$. Based on the Theorem \ref{Theorem:REEU}, we can know that $u(r^\prime, s)$ $\le$ \textit{REEU}$(r,s)$ $\le$ \textit{REEU}$(\alpha,s)$ = \textit{RERSU}$(r,s)$. Therefore, we have $u(r^\prime)$ $\le$ \textit{RERSU}$(r)$.
\end{proof}

\begin{strategy}[Left expansion estimated utility pruning strategy]
	\label{Strategy:LEEU}
	\rm According to Theorem \ref{Theorem:LEEU}, in the left expansion of a sequential rule $r$, US-Rule proposes the left expansion estimated utility pruning strategy (LEEUP) to safety prune $r$ which satisfies \textit{LEEU}$(r)$ $\textless$ \textit{minutil}. We can know that $u(r^\prime)$ $\le$ \textit{LEEU}$(r)$ $\textless$ \textit{minutil}, the expanded  sequential rule $r^\prime$ is an unpromising rule.
\end{strategy}

\begin{strategy}[Right expansion estimated utility pruning strategy]
	\label{Strategy:REEU}
	\rm According to Theorem \ref{Theorem:REEU}, in the right expansion of a sequential rule $r$, US-Rule proposes the right expansion estimated utility pruning strategy (REEUP) to safety prune $r$ which satisfies \textit{REEU}$(r)$ $\textless$ \textit{minutil}. We can know that $u(r^\prime)$ $\le$ \textit{REEU}$(r)$ $\textless$ \textit{minutil}, the expanded  sequential rule $r^\prime$ is an unpromising rule.
\end{strategy}

\begin{strategy}[Left expansion reduced sequence utility pruning strategy]
	\label{Strategy:LERSU}
	\rm According to Theorem \ref{Theorem:LERSU}, in the left expansion of a sequential rule $r$, US-Rule proposes left expansion reduced sequence utility pruning strategy (LERSUP) to safety prune $r$ which satisfies \textit{LERSU}$(r)$ $\textless$ \textit{minutil}. We can know that $u(r^\prime)$ $\le$ \textit{LERSU}$(r)$ $\textless$ \textit{minutil}, the expanded  sequential rule $r^\prime$ is an unpromising rule.
\end{strategy}

\begin{strategy}[Right expansion reduced sequence utility pruning strategy]
	\label{Strategy:RERSU}
	\rm According to Theorem \ref{Theorem:RERSU}, in the left expansion of a sequential rule $r$, US-Rule proposes left expansion reduced sequence utility pruning strategy (RERSUP) to safety prune $r$ which satisfies \textit{RERSU}$(r)$ $\textless$ \textit{minutil}. We can know that $u(r^\prime)$ $\le$ \textit{RERSU}$(r)$ $\textless$ \textit{minutil}, the expanded  sequential rule $r^\prime$ is an unpromising rule.
\end{strategy}

\begin{strategy}[Rule estimated utility recomputing pruning strategy]
	\rm Strategy \ref{Strategy:UIP} can remove some unpromising items to avoid low-utility sequential rules are generated. However, in some sparse databases, there are also many items may not satisfy the condition \textit{SEU}$(i)$ $\le$ \textit{minutil} when removing unpromising items once. Because when removing unpromising items once, the SEU of some items may become smaller. To address this issue, US-Rule proposes rule estimated utility recomputing pruning strategy (REURP) to recalculate the SEU of all items multiple times and remove unpromising items.
\end{strategy}

\subsection{Data structures for novel upper bounds}

For efficiency and less memory consumption, we design the corresponding data structures for the introduced upper bounds (LEEU, REEU, LERSU, and RERSU). The details are described below.

\begin{definition}[RE-table of a sequential rule]
	\rm The RE-table is designed for the expansion of sequential rule $r$. In a sequence database $\mathcal{D}$, its format can presented as a tuple, like (\textit{SID}, \textit{Utility}, \textit{LEEU}, \textit{REEU}, \textit{RE-elements}). \textit{SID} is the unique identifier of sequence which contains $r$, and can also be defined as $\forall \textit{SID} \in \textit{seq}(r)$. \textit{Utility} is the total utility of this sequential rule $r$. \textit{LEEU} and \textit{REEU} are corresponding upper bound value of $r$. By storing both \textit{LEEU} and \textit{REEU} values, US-Rule can quickly utilize Strategy \ref{Strategy:LEEU} and Strategy \ref{Strategy:REEU}, thus avoiding the associated computation. As for \textit{RE-elements}, it is a list of \textit{RE-elements}. \textit{RE-elements} records the related information of a sequence containing $r$. We can see that \textit{RE-elements} is presented as (\textit{SID}, \textit{Utility}, \textit{ULeft}, \textit{URight}, \textit{ULeftRight}, \textit{LEEU}, \textit{REEU}, \textit{Position}). In \textit{RE-elements}, the first seven variables have been mentioned in the previous content. \textit{Position} is used to store the position of related itemset for efficiency. 
\end{definition}

\begin{definition}[LE-table of a sequential rule]
	\rm Since we specify that the right expansion cannot be performed after left expansion, then \textit{URight} and \textit{REEU} will carry little impact. Therefore, we designed \textit{LE-table} for better memory saving. \textit{LE-table} can presented as a tuple like (\textit{SID}, \textit{Utility}, \textit{LEEU}, \textit{LE-elements}). US-Rule just needs to use \textit{LEEU} to leverage the Strategy \ref{Strategy:LEEU}. \textit{LE-elements} can denoted as (\textit{SID}, \textit{Utility}, \textit{ULeft}, \textit{LEEU}), US-Rule has simplified it equally. Note that \textit{ULeft} is a little different from the above illustration. It is the sum of previous \textit{ULeft} and previous \textit{ULeftRight}. US-Rule uses one variable to optimize memory consumption.
\end{definition}

Note that RE-table and LE-table are modified from HUSRM \cite{zida2015efficient}.

\begin{definition}[RSU-table of a sequential rule]
	\label{RSU-table}
	\rm Unlike \textit{LEEU} and \textit{REEU}, US-Rule uses a \textit{RUS-table} to store and update \textit{LERSU} and \textit{RERSU} in left expansion or right expansion. The \textit{RUS-table} is a hash-table and US-Rule will keep updating it in the expansion. For each expanded item, US-Rule updates the \textit{RUS-table} once every time it scans a sequence containing it. If the sum of the value of the expanded item in the RUS-table and the remaining \textit{LERSU} or remaining \textit{RERSU} is less than \textit{minutil} during the update process, US-Rule stops the operation associated with the expanded item by using Strategy \ref{Strategy:LERSU} and Strategy \ref{Strategy:RERSU}.
\end{definition}

\subsection{Proposed US-Rule algorithm}

Based on the aforementioned discussions, we propose the US-Rule algorithm in this section. US-Rule first scans the sequence database to generate all sequential rules of $1$ $\ast$ $1$. And then, it uses a depth-first search to expand all $1$ $\ast$ $1$ rules. As discussed in Definition \ref{ESR}, to ensure that no unnecessary sequential rule is generated, a sequential rule cannot perform the right expansion after it has performed the left expansion. There are two procedures in US-Rule, one is the \textit{leftExpansion} procedure, and another is the  \textit{rightExpansion} procedure. The pseudocode of US-Rule is shown in Algorithm \ref{alg:USRule}.

\begin{algorithm}[htbp]
	\caption{The US-Rule algorithm}
	\label{alg:USRule}
	\KwIn{$\mathcal{D}$: a sequence database, \textit{minconf}: the minimum confidence threshold defined by users, \textit{minutil}: the minimum utility threshold defined by users.}
	\KwOut{a set of all high-utility sequential rules.}
	
	initialize $I$ to get all items occurred in $\mathcal{D}$;
	
	\While {$\exists$ $i$ $\in$ $I$ and $\textit{SEU}(i)$ $\textless$ \textit{minutil}} {
		remove all unpromising items from $I$; \\
		update the SEU of all items in $I$; \qquad(\textbf{REURP Strategy} );
	}
	
	scan $\mathcal{D}$ to build \textit{REUCM} and $R$ (a set of rules of the form $r$ $:$ $a$ $\rightarrow$ $b$ $(a,b \in I))$; \\
	calculate the SEU and \textit{seq}$(r)$ of $r$; \\
	remove unpromising sequential rules from $R$;  \qquad(\textbf{USRP Strategy}); 
	
	\For {$r$ $\in$ $R$}{
		calculate $u(r)$ by scanning \textit{seq}$(r)$; \\
		\If {$u(r)$ $\ge$ \textit{minutil} and \textit{conf}$(r)$ $\ge$ \textit{minconf}}{
			output $r$;
		}
		\If {\textit{REEU}$(r)$ $\ge$ \textit{minutil}}{
			call \textbf{rightExpansion}$($$r$, \textit{REUCM}, $\mathcal{D}$, \textit{minutil}, \textit{minconf}$)$; \qquad(\textbf{REEUP Strategy});
		}
		\If {\textit{LEEU}$(r)$ $\ge$ \textit{minutil}}{
			call \textbf{leftExpansion}$($$r$, \textit{REUCM}, $\mathcal{D}$, \textit{minutil}, \textit{minconf}$)$; \qquad(\textbf{LEEUP Strategy}); 
		}
	}	
	
\end{algorithm}

\begin{algorithm}[htbp]
	\caption{The rightExpansion procedure}
	\label{alg:rightExpansion}
	
	\KwIn{$r$: a sequential rule $\{X\}$ $\rightarrow$ $\{Y\}$, \textit{REUCM}: the rule estimated utility con-occurrence map, $\mathcal{D}$: a sequence database, \textit{minconf}: the minimum confidence threshold defined by users, \textit{minutil}: the minimum utility threshold defined by users.}
	\KwOut{a set of all high-utility sequential rules.}
	
	initialize \textit{RSU-table} $\leftarrow$ $\emptyset$, \textit{RE-table} $\leftarrow$ $\emptyset$, \textit{rules} $\leftarrow$ $\emptyset$;
	
	\For {$r$ $\in$ \textit{R} } {
		initialize \textit{RRSU} $\leftarrow$ \textit{RERSU}$(r)$, $N$ $\leftarrow$ the largest item in the antecedent of $r$; \\
		\For {sequence $s$ whose \textit{sid} $\in$ \textit{seq}$(r)$} {
			initialize \textit{RRSU} $\leftarrow$ $\emptyset$; \\
			\For {$i$ $\in$ $s$ which can expanded by right expansion} {
				\If {\textit{REUCM}$(N,i)$ $\textless$ \textit{minutil}} {
					continue;  \qquad(\textbf{REUCP Strategy}); 
				}
				$t$ $\leftarrow$ $($$\{X\}$ $\rightarrow$ $\{Y\} \cup i$$)$; \\
				update \textit{RSU-table}$(i)$; \\
				\If {\textit{RRSU} + \textit{RSU-table}$(i)$ $\textless$ \textit{minutil}} {
					\If {$t$ $\in$ \textit{rules}} {
						remove $t$ from \textit{rules}; \\
						continue;   \qquad(\textbf{RERSUP Strategy}); 
					}
				}
				update \textit{rules}, \textit{RE-table} of $t$;
			}
		}
		
	}
	
	\For {sequential rule $t$ $\in$ \textit{rules}} {
		calculate $u(t)$ by scanning the \textit{RE-table} of $t$; \\
		\If {$u(t)$ $\ge$ \textit{minutil} and \textit{conf}$(t)$ $\ge$ \textit{minconf}} {
			output $t$;
			
		}
		\If {\textit{REEU}$(t)$ $\ge$ \textit{minutil}} {
			call \textbf{rightExpansion}$($$t$, \textit{REUCM}, $\mathcal{D}$, \textit{minutil}, \textit{minconf}$)$; \qquad(\textbf{REEUP Strategy}); 
		}
		\If {\textit{LEEU}$(t)$ $\ge$ \textit{minutil}} {
			call \textbf{leftExpansion}$($$t$, \textit{REUCM}, $\mathcal{D}$, \textit{minutil}, \textit{minconf}$)$; \qquad(\textbf{LEEUP Strategy});
		}
	}
	
\end{algorithm}

US-Rule first scans the database to access all items which appear in this database and calculates the SEU of these items (Line 1). According to the REUCP, US-Rule continuously removes unpromising items from the database and updates the SEU of the items that have not been removed (Lines 2-5). US-Rule then scans the database again, generates the set of all $1$ $\ast$ $1$ sequential rules, and calculates their corresponding SEU values (Lines 6-7). According to the USRP, US-Rule removes unpromising rules from $R$ (the set of all $1$ $\ast$ $1$ sequential rules) (Line 8). Subsequently, US-Rule traverses all sequential rules $r$ in $R$. And then, after calculating its utility and confidence, if both its \textit{minutil} and \textit{minconf} are satisfied, sequential rule $r$ is a high-utility sequential rule and US-Rule will output it (Lines 10-13). Immediately afterwards, according to LEEUP and REEUP, US-Rule determines whether the REEU and LEEU of $r$ are greater than \textit{minutil}. If this condition is satisfied, the corresponding left extension and right extension are recursively executed (Lines 14-19). Obviously, in the main program of the US-Rule algorithm, we can see that for the use of the pruning strategies for the new upper bounds, US-Rule uses only the depth pruning strategy (REEUP and LEEUP). And in the two expansions, the width pruning strategy and the depth pruning strategy will be used. Note that REUCP is also used in these two expansions. The pseudocode for the \textit{rightExpansion} and the \textit{leftExpansion} is shown in Algorithm \ref{alg:rightExpansion} and Algorithm \ref{alg:leftExpansion} respectively.

In \textit{rightExpansion}, US-Rule first initializes the \textit{RSU-table}, \textit{RE-table} and \textit{rules} to the empty set (Line 1). As discussed in Definition \ref{RSU-table}, the \textit{RSU-table} is used to utilize RERSUP and LERSUP. Then the \textit{RE-table} and \textit{rules} are used to store information related to the expanded sequential rules. Subsequently, US-Rule scans the sequences containing the rules in $R$. The RRSU of $r$ and the largest item $N$ in the antecedent of $r$ are obtained (Lines 2-5). If the expanded item $i$ does not satisfy $\textit{REUCM}(N, i)$ greater than \textit{minutil}, according to REUCP, US-Rule can stop the expansion of $r$ with $i$ (Lines 7-9). If this condition is satisfied, US-Rule lets $t$ be the expanded sequential rule and updates the \textit{RSU-table} (Lines 10-11). And then, US-Rule determines whether $\textit{RSU}(t)$ satisfies \textit{minutil}, if not, US-Rule can use RERSUP to stop the operation associated with it and remove $t$ from \textit{rules} (Lines 12-17). Otherwise, US-Rule updates \textit{rules} and the corresponding \textit{RE-table} of $t$ (Line 18). Subsequently, US-Rule repeats the session similar to the main program. Calculating the utility value of $t$ by scanning its \textit{RE-table} (Line 23). If it is a high-utility sequential rule, US-Rule outputs it (Lines 24-26). Then the REEU and LEEU of $t$ are compared to \textit{minutil} and US-Rule determines whether to continue expansion based on LEEUP and REEUP (Lines 27-32).

The procedure of \textit{leftExpansion} is similar to the \textit{rightExpansion}. There is a little different from the \textit{rightExpansion}. We can consider that the procedure of the \textit{leftExpansion} is the \textit{rightExpansion} without consequent growth. Therefore, we do not discuss the more details about it.

\begin{algorithm}[htbp]
	\caption{The leftExpansion procedure}
	\label{alg:leftExpansion}
	
	\KwIn{$r$: a sequential rule $\{X\}$ $\rightarrow$ $\{Y\}$, \textit{REUCM}: the rule estimated utility con-occurrence map, $\mathcal{D}$: a sequence database, \textit{minconf}: the minimum confidence threshold defined by users, \textit{minutil}: the minimum utility threshold defined by users.}
	\KwOut{a set of all high-utility sequential rules.}
	
	initialize \textit{RSU-table} $\leftarrow$ $\emptyset$, \textit{LE-table} $\leftarrow$ $\emptyset$, \textit{rules} $\leftarrow$ $\emptyset$;
	
	\For {$r$ $\in$ \textit{R} } {
		initialize \textit{RRSU} $\leftarrow$ \textit{LERSU}$(r)$, $M$ $\leftarrow$ the largest item in the consequent of $r$; \\
		\For {sequence $s$ whose \textit{sid} $\in$ \textit{seq}$(r)$} {
			initialize \textit{RRSU} $\leftarrow$ $\emptyset$; \\
			\For {$i$ $\in$ $s$ which can expanded by left expansion} {
				\If {\textit{REUCM}$(i,M)$ $\textless$ \textit{minutil}} {
					continue;  \qquad(\textbf{REUCP Strategy}); 
				}
				$t$ $\leftarrow$ $($$\{X\} \cup i$ $\rightarrow$ $\{Y\} $$)$; \\
				update \textit{RSU-table}$(i)$; \\
				\If {\textit{RRSU} + \textit{RSU-table}$(i)$ $\textless$ \textit{minutil}} {
					\If {$t$ $\in$ \textit{rules}} {
						remove $t$ from \textit{rules}; \\
						continue;   \qquad(\textbf{LERSUP Strategy}); 
					}
				}
				update \textit{rules}, \textit{LE-table} of $t$;
			}
		}
		
	}
	
	\For {sequential rule $t$ $\in$ \textit{rules}} {
		calculate $u(t)$ by scanning the \textit{LE-table} of $t$; \\
		\If {$u(t)$ $\ge$ \textit{minutil} and \textit{conf}$(t)$ $\ge$ \textit{minconf}} {
			output $t$;
			
		}
		\If {\textit{LEEU}$(t)$ $\ge$ \textit{minutil}} {
			call \textbf{leftExpansion}$($$t$, \textit{REUCM}, $\mathcal{D}$, \textit{minutil}, \textit{minconf}$)$; \qquad(\textbf{LEEUP Strategy});
		}
	}
	
\end{algorithm}

\subsection{Example of US-Rule}

In this section, we introduce a simple running example of US-Rule. The sequence database is the same as shown in Table \ref{table1}, and the external utilities are shown in Table \ref{table2}. \textit{minutil} and \textit{minconf} are set to $10$ and $0.5$ respectively. At first, US-Rule scans the database to find out all items appearing in the database and calculate their SEU values, as shown in Table \ref{SEUTable0}. It is clear that $f$ is an unpromising item,  because \textit{SEU}$(f)$ = 5 $\textless$ \textit{minutil}. US-Rule will remove $f$ from the database. And then, the SEU of some items may be changed. The SEU values of all items after one removal are shown in Table \ref{SEUTable1}. We can find that $e$ is not satisfied the \textit{minutil} after its SEU was updated. It should be removed from the database. The REURP is helpful to eliminate low-utility items. After multiple removal, the final database is shown in Table \ref{FinalDT} and the SEU table is shown in Table \ref{SEUTable2}. After that, US-Rule uses these items to generate all 1 $\ast$ 1 sequential rules, as shown in Table \ref{ALL 1SR}. The REUCM is constructed simultaneously. We then know that \textit{REUCP}$(a,c)$ = 16, \textit{REUCP}$(a,g)$ = 16, \textit{REUCP}$(b,d)$ = 10, \textit{REUCP}$(b,g)$ = 19, \textit{REUCP}$(c,g)$ = 17, and \textit{REUCP}$(d,g)$ = 10. According to USRP, total three sequential rules will be eliminated. Let us discuss $r$ = $a$ $\rightarrow$ $c$ firstly. The LEEU and REEU of $r$ are equal to 16, thus US-Rule can continue right expansion and left expansion. In the right expansion of $r$, only item $g$ can be expanded to a new sequential rue $t$. \textit{REUCM}$(a,g)$ and \textit{RERSU}$(t)$ satisfy \textit{minutil}. Therefore, it will obtain $t$. $t$ = $\{a\}$ $\rightarrow$ $\{c,g\}$, its confidence is 1.0 and its utility is 14. We can consider $t$ as a HUSR. In the left expansion of $r$, only item $b$ can be expanded to a new sequential rue $t$. However, \textit{REUCM}$(b,c)$ = 9 $\textless$ \textit{minutil}. According to REUCP, $\{a,b\}$ $\rightarrow$ $\{c\}$ is not a HUSR. The process of other 1 $\ast$ 1 sequential rules are similar. The final results of HUSRs are shown in Table \ref{table3}.

\begin{minipage}{\textwidth}
	\begin{minipage}[t]{0.45\textwidth}
		\centering
		\makeatletter\def\@captype{table}\makeatother
		\begin{minipage}[t]{\textwidth}
			\centering
			\makeatletter\def\@captype{table}\makeatother\setlength{\abovecaptionskip}{5pt}\setlength{\belowcaptionskip}{5pt}\caption{SEU table of all items}
			\label{SEUTable0}
			\begin{tabular}{|c|c|c|c|c|c|c|c|}
				\hline
				\textbf{Item}	    & $a$	& $b$	& $c$	& $d$	& $e$	& $f$  & $g$ \\ \hline 
				\textbf{SEU}	    & 17     & 19     & 17     & 10     & 13     & 5    & 27 \\ \hline			
			\end{tabular}
		\end{minipage}
		\begin{minipage}[t]{\textwidth}
			\centering
			\makeatletter\def\@captype{table}\makeatother\setlength{\abovecaptionskip}{5pt}\setlength{\belowcaptionskip}{10pt}\caption{SEU table of all items after one removal}     
			\label{SEUTable1}  
			\begin{tabular}{|c|c|c|c|c|c|c|}
				\hline
				\textbf{Item}	    & $a$	& $b$	& $c$	& $d$	& $e$  & $g$ \\ \hline 
				\textbf{SEU}	    & 17     & 19     & 17     & 10     & 9   & 27 \\ \hline
			\end{tabular}
		\end{minipage}
		\begin{minipage}[t]{\textwidth}
			\centering
			\makeatletter\def\@captype{table}\makeatother\setlength{\abovecaptionskip}{5pt}\setlength{\belowcaptionskip}{10pt}\caption{SEU table of all items after multiple removal}
			\label{SEUTable2}
			\begin{tabular}{|c|c|c|c|c|c|c|}
				\hline
				\textbf{Item}	    & $a$	& $b$	& $c$	& $d$  & $g$ \\ \hline 
				\textbf{SEU}	    & 16     & 19     & 16     & 10     & 26 \\ \hline
			\end{tabular}
		\end{minipage}
	\end{minipage}
	\begin{minipage}[t]{0.45\textwidth}
		\begin{minipage}[t]{\textwidth}
			\centering
			\makeatletter\def\@captype{table}\makeatother\setlength{\abovecaptionskip}{5pt}\setlength{\belowcaptionskip}{5pt}\caption{Sequence database after multiple removal}
			\label{FinalDT}
			\begin{tabular}{|c|c|}  
				\hline 
				\textbf{SID} & \textbf{Sequence} \\
				\hline  
				\(s_{1}\) & $<$$(a, 1)$ $(b, 2)$ $\{(c, 1) (g,1)\}$$>$ \\ 
				\hline
				\(s_{2}\) & $<$$(a, 1)(c, 1)$ $(g, 1)$$>$ \\  
				\hline  
				\(s_{3}\) & $<$$(b, 1)$ $(d, 1)$ $(g, 2)$$>$ \\
				\hline  
			\end{tabular}
		\end{minipage}
	\end{minipage}
\end{minipage}

\begin{table}[!htbp]
	\centering
	\caption{All $1$ $\ast$ $1$ sequential rules found by US-Rule}
	\label{ALL 1SR}
	\begin{tabular}{|c|c|c|c|c|}  
		\hline 
		\textbf{ID} & \textbf{Sequential rule}  & \textbf{Confidence} & \textbf{Utility} & \textbf{Eliminated} \\
		\hline 
		\(r_{1}\) & $\{a\}$ $ \rightarrow $ $\{b\}$ & 0.5 & 9 & Yes\\ 
		\hline 
		\(r_{2}\) & $\{a\}$ $ \rightarrow $ $\{c\}$ & 1.0 & 16 & Not\\ 
		\hline 
		\(r_{3}\) & $\{a\}$ $ \rightarrow $ $\{g\}$ & 1.0 & 16 & Not\\ 
		\hline
		\(r_{4}\) & $\{b\}$ $ \rightarrow $ $\{c\}$ & 0.5 & 9 & Yes\\ 
		\hline  
		\(r_{5}\) & $\{b\}$ $ \rightarrow $ $\{d\}$ & 1.0 & 10 & Not\\ 
		\hline
		\(r_{6}\) & $\{b\}$ $ \rightarrow $ $\{g\}$ & 1.0 & 19 & Not\\ 
		\hline
		\(r_{7}\) & $\{c\}$ $ \rightarrow $ $\{g\}$ & 0.5 & 7 & Yes\\ 
		\hline
		\(r_{8}\) & $\{d\}$ $ \rightarrow $ $\{g\}$ & 1.0 & 10 & Not\\ 
		\hline
	\end{tabular}
\end{table}

\section{Experiments} 
\label{sec:experiments}

In this section, in order to evaluate the performance of US-Rule algorithm, we perform experiments on several real-life datasets and synthetic datasets. Considering that different strategies have different effects on different datasets, we designed four versions of the US-Rule algorithm to compare the state-of-the-art HUSRM algorithm \cite{zida2015efficient}. We use US-Rule$_{V1}$, US-Rule$_{V2}$, US-Rule$_{V3}$, and US-Rule$_{V4}$ to denote the four versions of the US-Rule algorithm. These four versions use the Strategy \ref{Strategy:UIP} and Strategy \ref{Strategy:USRP}. In addition, US-Rule$_{V1}$ only uses REUCP, US-Rule$_{V2}$ uses REUCP and four pruning strategies (including LEEUP, REEUP, LERSUP, and RERSUP). US-Rule$_{V3}$ uses four pruning strategies and REURP. As for US-Rule$_{V4}$, it uses all pruning strategies including REUCP, four pruning strategies, and REURP.

We conducted our experiments in a PC equipped with a 64-bit Windows 10 operating system, 3.6 GHz AMD Ryzen 5 3600 CPU, and 8 GB RAM. Both HUSRM and US-Rule algorithms are implemented in Java language. The details of the experiment are given below.

\subsection{Datasets for the experiment}

To assess the efficiency and scalability of the US-Rule algorithms on different characteristic datasets, we use six datasets. There are four real-life datasets and two datasets. In our experiment, the real-life datasets are generated from book, click-stream, and sign language. They are generated from linguistic datasets, and they can be obtained and transformed from a page or a book. It can be considered that we test a single item sequence in real datasets and multiple item sequences in synthetic datasets. We use bible, kosarak, leviathan, and sign as four real-life datasets which can be obtained from SPMF\footnote{\url{http://www.philippe-fournier-viger.com/spmf/}}. These include sparse datasets, dense datasets, and long sequence datasets, allowing the experiments to evaluate the algorithms from different types of data. As for synthetic datasets, we can consider that they are generated from IBM data generator \cite{agrawal1995mining}. The details of all experimental datasets are shown in Table \ref{DATA} and the description of them are given as following.

\begin{table}[!h]
	\caption{Details of different experimental datasets}  
	\label{DATA}
	\centering
	\begin{tabular}{|c|c|c|c|c|c|}
		\hline
		\textbf{Dataset} & \textbf{$\vert \mathcal{D} \vert$} & \textbf{$\vert \textit{I} \vert$} & \textbf{AvgLen} & \textbf{AvgEle} &  \textbf{Type} \\ \hline   \hline   
		
		Bible & 36369 & 13905	& 21.64 & 1.0  & Moderately dense, slightly long \\ \hline
		Kosarak10k	& 10000 & 10094	& 8.14 & 1.0  & Sparse, short \\ \hline
		Leviathan & 5834 & 9025	& 33.81 & 1.0  & Slightly dense, moderately long \\ \hline
		Sign	& 730 & 267	& 51.99 & 1.0  & Dense, long \\ \hline
		Syn10k	& 10000 & 7312	& 27.11 & 4.35  & Moderately dense, moderately long \\ \hline
		Syn20k	& 20000 & 7442	& 26.97 & 4.33  & Moderately dense, moderately long \\ \hline
		
	\end{tabular}
\end{table}

$ \bullet $ \textbf{Bible} is a moderately dense dataset converted from the Bible. It is also a slightly long sequence and one item dataset, with each sequence corresponding to a sentence of the Bible.

$ \bullet $ \textbf{Kosarak} is a particularly sparse short sequence dataset of click-stream from a Hungarian online news website. Likewise, it  is also an item-based dataset.

$ \bullet $ \textbf{Leviathan} is a dataset generated from the novel Leviathan written by Tomas Hobbes. It is a slightly dense and moderately long sequence dataset, and each item contains one item.

$ \bullet $ \textbf{Sign} is a sign language and item-based dataset. It is the densest in our selection of real datasets, with many long sequences and a few hundred items. 

$ \bullet $ \textbf{Syn10k} is a synthetic dataset generated from IBM data generator. It has 10,000 sequences and many moderately long sequences. Different from real datasets, it is a multiple-items-based dataset.

$ \bullet $ \textbf{Syn20k} is also a synthetic dataset generated from IBM data generator. It has 20,000 sequences. Different from the Syn10k, the number of different items is increased a little.

\subsection{Efficiency analysis}

In this subsection, we perform a number of experiments in terms of the running time. Four versions are used to compare the state-of-the-art HUSRM. When setting \textit{minconf} to 0.6 and setting various \textit{minutil} on different datasets, the experimental results are shown in Fig. \ref{Runtime}.

\begin{figure}[ht]
	\centering
	\includegraphics[clip,scale=0.07]{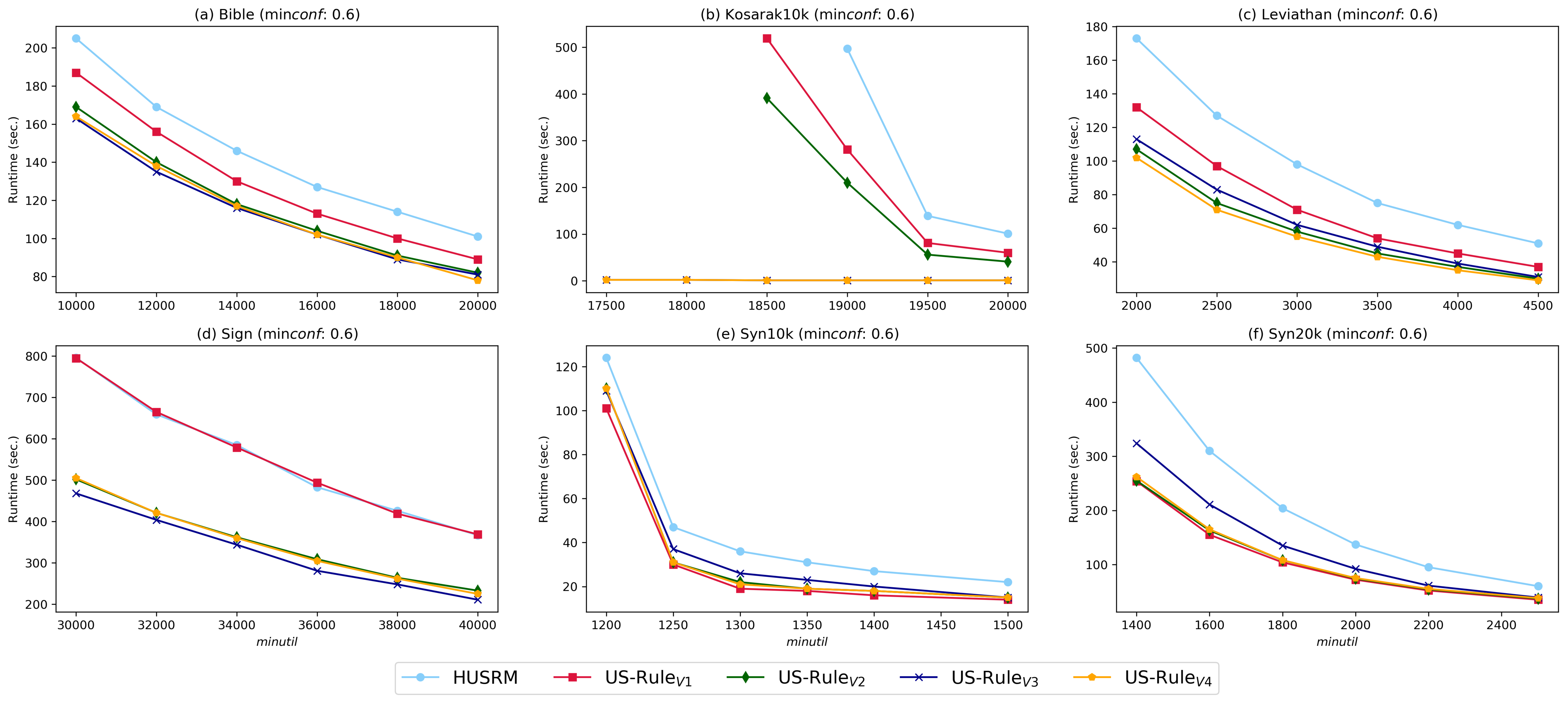}
	\caption{Runtime under various \textit{minutil}}
	\label{Runtime}
\end{figure}

 From the results, we can see that all four variants of US-Rule are better than HUSRM, especially on dense and long sequence datasets. On moderately datasets, such as Bible, Leviathan, and Synthetic datasets, we can clearly see that both REUCP and upper bounds pruning strategies can be effective. US-Rule$_{V4}$ is the fastest algorithm on datasets Bible and Leviathan. The difference between US-Rule$_{V2}$, US-Rule$_{V3}$, and US-Rule$_{V4}$ is not significant on Bible. While on Leviathan, we can see that as the number of optimization strategies increases, the running efficiency improves. On the two synthetic datasets, the running times of the three variants US-Rule$_{V1}$, US-Rule$_{V2}$ and US-Rule$_{V4}$ are almost the same. US-Rule$_{V1}$ is the fastest algorithm on datasets Syn10k and Syn20k. The reason is that the effect of REUCP is very obvious than the  pruning strategies on synthetic datasets. Although pruning strategies are powerful, they are not as effective as REUCP. On the Kosarack, a sparse and short sequence dataset, US-Rule$_{V3}$ and US-Rule$_{V4}$ can achieve a significant performance. REURP can quickly filter many unpromising items after the re-calculation of SEU. Among HUSRM, US-Rule$_{V1}$ and US-Rule$_{V2}$, there is also a gap in their running time, and usually HUSRM is the slowest one. As \textit{minutil} decreases, these algorithms can not discover HUSRs from datasets within an acceptable running time. On the dense dataset Sign, we can see that the running time of US-Rule$_{V1}$ and HUSRM are basically the same. This is consistent with our expectation that RECUP does not reduce much invalid computation on dense datasets. As for US-Rule$_{V2}$ and US-Rule$_{V3}$, their running time is almost the same. US-Rule$_{V4}$, which uses REURP, further optimizes the runtime. In summary, we can draw the conclusion as follows: on different types of datasets, although there are performance differences between the variants of US-Rule algorithm using different strategies, they are always more efficient than HUSRM.

\subsection{Memory evaluation}

The experimental results for the usage of memory are shown in Fig. \ref{memory}. On the Bible dataset, we can see that the memory consumed by HUSRM, US-Rule$_{V2}$, and US-Rule$_{V4}$ is not very different. For US-Rule$_{V3}$, which does not use SEUCP, it consumes the least amount of memory. While US-Rule$_{V1}$, which only uses SEUCP, consumes more memory. On Kosarack10k, since REURP filters many unpromising items, it makes the memory consumption of US-Rule$_{V3}$ and US-Rule$_{V4}$ particularly low. For US-Rule$_{V1}$ and US-Rule$_{V2}$, the use of REUCM increases the memory consumption. On the Leviathan dataset, the memory consumption fluctuates a bit. We can still see that it is similar to the case on the Bible dataset, i.e., US-Rule$_{V3}$ generally uses less amount of memory. On the Sign dataset, the memory of each algorithm does not differ much, which may be due to the very small number of items in Sign. On the first synthetic dataset, we can see that HUSRM and US-Rule$_{V3}$ use more memory. REUCP works well on this dataset. In contrast, on the second synthetic dataset, the algorithm which uses REUCP usually consumes more memory. This is exactly as expected, since the second synthetic dataset is denser, making REUCP less effective.

\begin{figure}[ht]
	\centering
	\includegraphics[clip,scale=0.07]{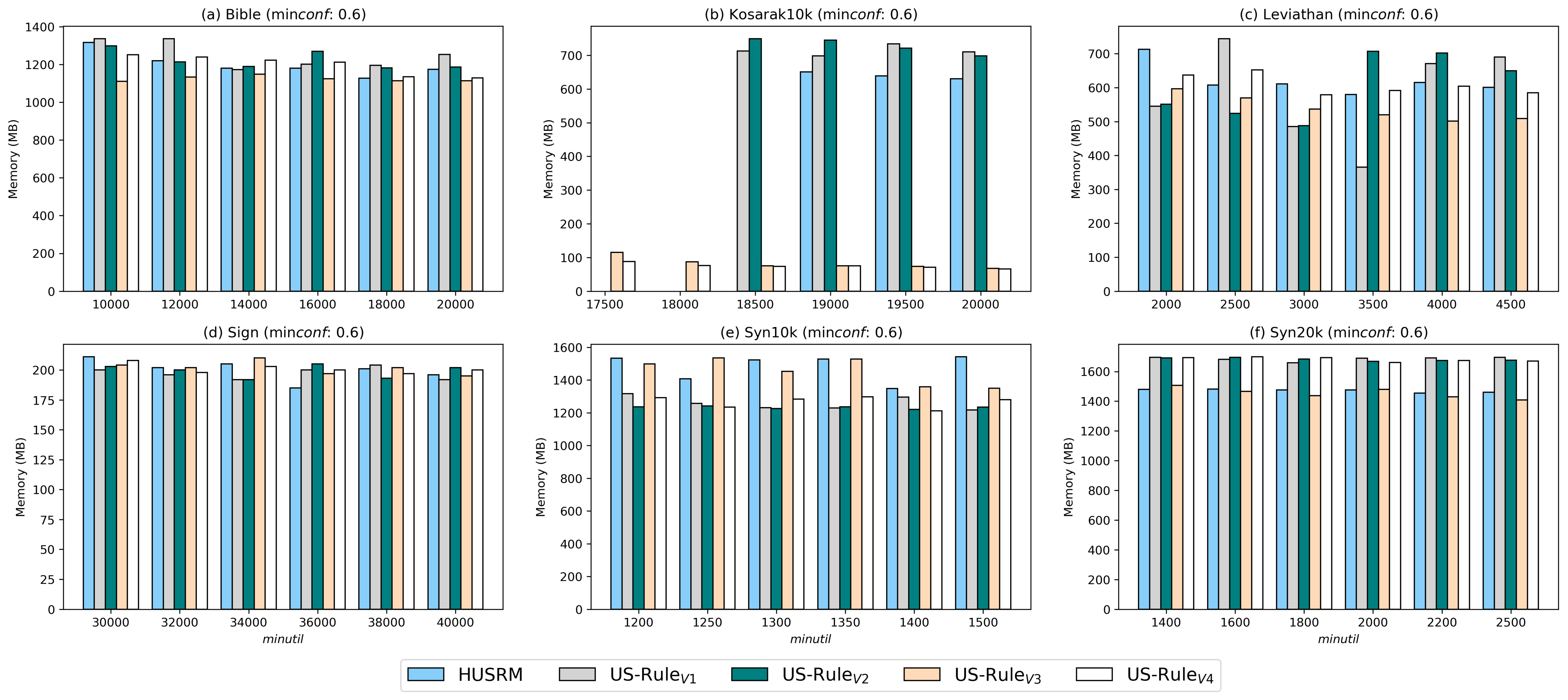}
	\caption{Memory usage under various \textit{minutil}}
	\label{memory}
\end{figure}

\subsection{Scalability}

In this subsection, to better compare the scalability of each algorithm, we use synthetic datasets and increase their size from 15k to 20k. We  evaluate these algorithms in terms of running time and memory consumption. The \textit{minconf} is set to 0.6 and the \textit{minutil} is set to 2000. The experimental results are shown in Fig. \ref{scalability}. Obviously, as the dataset size increased, both the running time and memory consumption of each algorithm are increased too. The experimental results are also expected, and the algorithm using the most optimized strategies can always achieve the best efficiency, but  consumes much memory. The algorithm using only the pruning strategies can improve the efficiency and does not cost much memory. In terms of runtime, although all algorithms show a linear increasing trend, all versions of US-Rule do not increase as much as that of HUSRM.

\begin{figure}[ht]
	\centering
	\includegraphics[clip,scale=0.07]{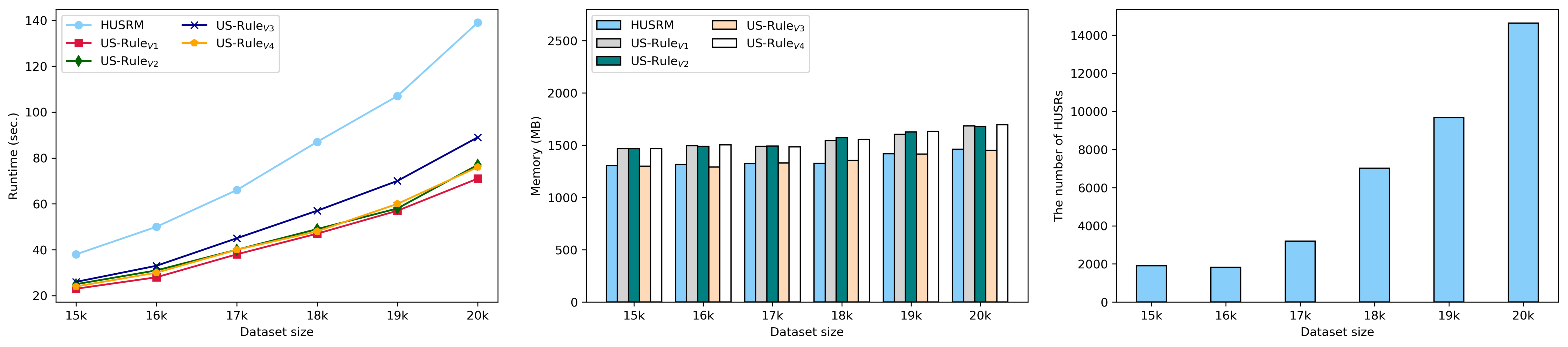}
	\caption{Scalability when \textit{minconf} = 0.6 and \textit{minutil} = 2000}
	\label{scalability}
\end{figure}

\section{Conclusions}  
\label{sec:conclusion}

In this paper, to improve the efficiency of high-utility sequential rule mining, we proposed a novel US-Rule algorithm. Based on REUCP, US-Rule can achieve a better performance. In order to solve the efficiency on dense or long sequence datasets, we introduced four tighter upper bounds and proposed their corresponding pruning strategies. To better address the issue on sparse and short sequence datasets, we further proposed REURP to filter unpromising items by recalculating the SEU of items multiple times. Finally, we conducted a large number of experiments on different datasets to evaluate different variants of the US-Rule algorithm. Combining the results of each experiment, we can say that US-Rule has better performance and scalability than the state-of-the-art HUSRM algorithm. Our future work may apply the US-Rule algorithm to deal with some applications in other fields, and also will adapt US-Rule with several constraints for more valuable sequence mining.

\section*{Acknowledgment}

This research was partially supported by National Natural Science Foundation of China (Grant No. 62002136), Guangdong Basic and Applied Basic Research Foundation, Guangzhou Basic and Applied Basic Research Foundation (Grant No. 202102020277).

\bibliographystyle{ACM-Reference-Format}
\bibliography{reference}


\end{document}